\documentclass[12pt,draftclsnofoot,onecolumn]{IEEEtran}

\usepackage [english]{babel}
\usepackage [autostyle, english = american]{csquotes}
\MakeOuterQuote{"}

\usepackage{amsmath,graphicx,amssymb,comment}
\usepackage{xpatch}
\usepackage{xcolor}
\usepackage{enumerate}
\usepackage{ntheorem}
\newtheorem{theorem}{Theorem}
\newtheorem{corollary}{Corollary}[theorem]

\newtheorem{prop}{Proposition}

\usepackage{stfloats}
\usepackage{caption} 
\usepackage{cite}
\usepackage{algorithm}
\usepackage{algorithmic}
\begin{document}             
\title{MMSE-Optimal Sequential Processing for Cell-Free Massive MIMO With Radio Stripes}
\author{
\medskip
\normalsize
Zakir Hussain Shaik,
{\em Student Member, IEEE},
Emil Bj\"{o}rnson,
{\em Senior Member, IEEE}, and \\
Erik G. Larsson,
{\em Fellow, IEEE}
\normalsize
\thanks{ Preliminary results on a precursor to proposed algorithm were presented in IEEE ICC 2020.
	The authors are with the Department of Electrical    Engineering (ISY), Linköping University, SE-58183 Linköping, Sweden. E. Bj\"ornson is also with the Department of Computer Science, KTH Royal Institute of Technology, SE-10044, Stockholm, Sweden. Emails: \{zakir.hussain.shaik, emil.bjornson, erik.g.larsson\}@liu.se.}
}



\maketitle
\begin{abstract}
Cell-free massive multiple-input-multiple-output (mMIMO) is an emerging technology for beyond 5G with its promising features such as higher spectral efficiency and superior spatial diversity as compared to conventional multiple-input-multiple-output (MIMO) technology. The main working principle of cell-free mMIMO is that many distributed access points (APs) cooperate simultaneously to serve all the users within the network without creating cell boundaries. This paper considers the uplink of a cell-free mMIMO system utilizing the radio stripe network architecture with a sequential fronthaul between the APs. A novel uplink sequential processing algorithm is developed which is proved to be optimal in both the maximum spectral efficiency (SE) and the minimum mean square error (MSE) sense. A detailed quantitative analysis of the fronthaul requirement or signaling of the proposed algorithm and its comparison with competing sub-optimal algorithms is provided. Key conclusions and implications are summarized in the form of corollaries. Based on the analytical and numerical simulation results, we conclude that the proposed scheme can greatly reduce the fronthaul signaling, without compromising the communication performance.
\end{abstract}

\begin{keywords}
\textnormal{Beyond 5G, radio stripes, cell-free massive MIMO, uplink, spectral efficiency, mean square error, sequential processing.
}
\end{keywords}

\IEEEpeerreviewmaketitle

\section{Introduction}
Massive multiple-input-multiple-output (mMIMO) networks, with their manifold benefits over conventional multiple-input-multiple-output (MIMO) networks \cite{LSAS,LS_MIMO,VLM_MIMO} such as high spatial resolution and very high-spectral efficiency (SE), have garnered intense interest in the past decade making it a reality in the year 2018 \cite{newNR,mMIMOreality}. Nevertheless, mMIMO in its original form suffers from large signal-to-noise ratio (SNR) variations between cell center and cell edge users. This problem can be tackled to some extent by small-cell networks \cite{cellFreeErik}, but again they suffer from high inter-cell interference due to its inherent cell-centric implementation. Hence, there's a need for a paradigm shift of networks from cellular to cell-free. This can be achieved by the more recently evolving idea of cell-free mMIMO network which is essentially a decentralized implementation of mMIMO \cite{cellFreeErik,cellFree2,centBjorn}.

A cell-free mMIMO network consists of a central processing unit (CPU) connected to a set of access points (APs) which jointly serve all the user equipments (UEs) in the network. An AP consists of antennas and the signal processing units required to operate them locally. The original idea of cell-free mMIMO networks was to have a dedicated fronthaul and power supply to every AP running up to the CPU (e.g., a star topology) \cite{cellFreeErik}. In the uplink, each AP after receiving the pilot and data signals, forwards them to the CPU, where final fusion of data is done, and the information signals are decoded. However, this network topology requires large fronthaul capacity from the APs to the CPU and for a wired implementation a long cable to be connected between each AP and the CPU, making the cost the main bottleneck for practical implementation. These factors limit its practicability and necessitates the search for more practical architectures that can decentralize this processing and reduce the fronthaul signaling.

In the literature, there have been various techniques and algorithms that are developed for decentralizing the signal processing in mMIMO systems \cite{JeonPartialDecent,Shirazinia2016Decent,Li2017Decent,Burr2018Cooperative,Sadeghi2014Distrub,Li2019Decent,Bassoy_CoMP} which can be adopted to use in cell-free mMIMO implementation. The network topologies of interest include sequential (daisy-chain network), centralized (star-like network), tree and fully connected (mesh network). The choice of a particular topology depends on the application of interest. One of the promising directions for cell-free networks is radio stripes, which utilizes the sequential topology\cite{interdonato2019ubiquitous,icc2020}. This architecture is suitable for deployments in dense areas such as sports arenas and railway stations with many APs and UEs per km$^2$, and large construction elements that the stripes can be attached to. This network comprises of sequentially connected APs in a daisy chain topology and shares the same cable for fronthaul and power supply, as illustrated in Fig. \ref{fig:SysModel}. The sequential implementation of cell-free mMIMO has the potential to deliver the benefits of mMIMO with much lower fronthaul requirements when compared to centralized. In this paper, we develop an optimal uplink processing algorithm for radio stripes in the sense of maximum SE and minimum mean square error (MSE).

\subsection{Related Work}
The decentralization methods found in the literature can be broadly categorized into two types: fully centralized and fully distributed implementation. In the former category, all the required processing is done at the CPU \cite{centBjorn,cellFreeErik,Shirazinia2016Decent},  while in the latter category, all the processing is done locally at the APs \cite{icc2020,JeonPartialDecent,rls2,Bertilsson}, {except for the final fusion at the CPU}, using statistical channel state information (CSI). A centralized implementation has superior performance because of its access to complete information but  the downside is the requirement of a very large fronthaul when compared to other topologies such as sequential processing. There are a few other methods which does not strictly fall into either of the above categories. The method proposed in \cite{JeonPartialDecent} is partially decentralizing the process i.e., the APs partially process the received signals and forward the Gramian of the channel matrix to the CPU where most of the signal processing and estimation of the signal is done. The authors have shown analytically that this strategy achieves the same performance as linear methods such as maximum ratio combining (MRC), zero forcing (ZF) and linear minimum mean square error (LMMSE). However, the authors have only presented the analysis for perfect CSI case.

The relevant works which focused on developing algorithms for a sequential network for mMIMO are \cite{rls2,icc2020,JeonPartialDecent}. The authors in \cite{rls2} proposed algorithms which decentralizes ZF in a sequential manner but only for the perfect CSI case. Moreover, SE analysis is not investigated in  \cite{JeonPartialDecent} and \cite{rls2}. 

\subsection{Contribution}
The main contributions of this paper are:
\begin{enumerate}[(i)]
\item We develop a novel uplink sequential processing algorithm which is proved to be optimal in both the SE and MSE sense.
\item We provide closed-form expressions for the SE and MSE of the proposed method and also for any sequential linear processing algorithm.
\item We prove that the ordering of the APs has no impact on the performance when using the proposed algorithm.
\item We provide update expressions for the SE and MSE when adding additional APs.
\item We quantify and compare the fronthaul requirements of the proposed algorithm with competing algorithms \cite{rls2,icc2020}.
\item To address the latency issue that appears in long sequential networks, we provide a semi-distributed algorithm without loss in performance.
\item We provide numerical and simulation results comparing the proposed algorithm with the existing algorithms for sequential and centralized cell-free networks.
\end{enumerate}

In the conference version \cite{icc2020}, we propose the algorithm that had a trade-off between the SE and fronthaul signaling when compared to a centralized implementation\cite{centBjorn} with LMMSE receiver. Ideally, we would like to achieve the performance of a centralized scheme receiver but most of the fully decentralized algorithms fall short of achieving the optimal SE. This paper proposes a new algorithm that achieves the same performance as the optimal centralized implementation, while reducing the fronthaul requirement, thus the tradeoff issue is resolved.

\subsection{Paper Outline}
The remainder of this paper is organized as follows. Section \ref{SystemModel} presents the system model for uplink cell-free networks with a sequential fronthaul for both payload transmission and channel estimation. Section \ref{SeqProc} presents most of the contributions including an optimal algorithm for sequential signal processing and closed-form expressions of the SE and MSE achieved by the proposed method. In Section \ref{fronthaul}, a quantitative analysis of the fronthaul capacity requirements of the proposed method and other competing algorithms is provided. Numerical results are presented and analyzed in Section \ref{numericalResults}. Finally, the main conclusions of this paper are presented in Section \ref{conclu}. Appendix includes the proof of the main theorem.

\subsection{Notations}
Boldface lowercase letters, $\mathbf{a}$, denote column vectors and boldface uppercase letters, $\mathbf{A}$, denote matrices. The superscripts $(\cdot)^*,~(\cdot)^T,$ and $(\cdot)^H$ denote the conjugate, transpose, and Hermitian transpose, respectively. The notation $\mathbf{I}_N$ represents the $N\times N$ identity matrix. The $(m,n)$th element of a matrix $\mathbf{A}$ is denoted by $[\mathbf{A}]_{mn}$. A block-diagonal matrix is represented by $\mathrm{diag}(\mathbf{A}_1,\cdots,\mathbf{A}_N)$ for square matrices $\mathbf{A}_1,\cdots,\mathbf{A}_N$. The absolute value of a scalar and $l_2$-norm of a vector are denoted by $\vert \cdot \vert$ and $\Vert \cdot \Vert$, respectively. The real value of a scalar is denoted by $\Re\{\cdot\}$. We denote the expectation and variance by $\mathbb{E}\{\cdot\}$ and $\mathrm{Var}\{\cdot\}$, respectively. We use $\mathbf{z} \sim \mathcal{CN}\left(\mathbf{0},\mathbf{C}\right)$ to denote a multi-variate circularly symmetric complex Gaussian random vector with zero mean and covariance matrix $\mathbf{C}$.

\begin{figure}[!htbp]
	\centering
	\includegraphics[width=0.6\textwidth]{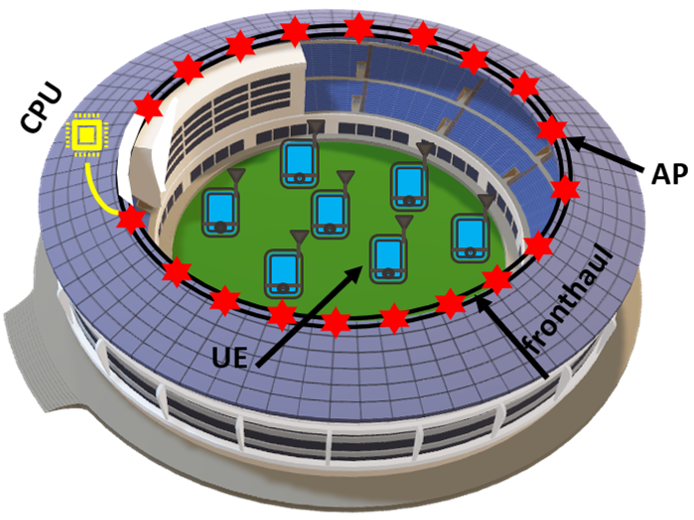}
	\caption{{Radio stripes network deployed over a football arena.}}
	\label{fig:SysModel}
	\vspace{-10mm}
\end{figure}

\section{System Model And Channel Estimation}\label{SystemModel}
We consider a cell-free mMIMO network comprising $L$ APs, each equipped with $N$ antennas. The fronthaul connections are assumed to go from  AP $1$ to AP $2$ $\cdots$ to AP $L$ to the CPU as shown in the Fig. \ref{fig:SysModel}. This architecture with a sequential fronthaul is called a radio stripe network. There are $K$ UEs, each with a single antenna, distributed arbitrarily in the network. We consider the standard block fading channel model with coherence block length of $\tau_c$ channel uses \cite{BjorScalableCellfree}. The channel between AP $l$ and UE $k$ is denoted by $\mathbf{h}_{kl} \in \mathbb{C}^{N}$. In each block, an independent realization is drawn from a correlated Rayleigh fading distribution as
\begin{equation}
\mathbf{h}_{kl} \sim \mathcal{CN} \left(\mathbf{0},\mathbf{R}_{kl}\right),
\end{equation}
where $\mathbf{R}_{kl} \in \mathbb{C}^{N \times N}$ is the spatial correlation matrix, which attributes the spatial channel correlation characteristics and large-scale fading. The large-scale fading coefficient describing the shadowing and pathloss is given by $\beta_{kl}\triangleq \mathrm{tr}\left(\mathbf{R}_{kl}\right)/N$. The spatial correlation matrices $\{\mathbf{R}_{kl}\}$ are assumed to be known at all the APs and the CPU.

This paper studies an uplink scenario where each coherence block consists of $\tau_p$ channel uses for pilot transmission to estimate the channels and $\tau_c - \tau_p$ channel uses for payload data. Both phases are described in detail below.

\subsection{Channel Estimation}
We assume there are $\tau_p$ mutually orthogonal $\tau_p$-length pilot signals $\boldsymbol{\phi}_1, \boldsymbol{\phi}_2,\ldots,\boldsymbol{\phi}_{\tau_p}$ with $\Vert \boldsymbol{\phi}_{k} \Vert^2=\tau_p$, which are used for channel estimation. We are mainly interested in the case $K>\tau_p$, where more than one UE is assigned the same pilot causing pilot contamination. We let the pilot assigned to UE $k$, for $ k=1,\ldots, K$, be denoted by $t_k \in \{1,\ldots,\tau_p\}$ and the set $\mathcal{S}_k = \{ i : t_i = t_k\}$ accounts for those UEs that are assigned the same pilot as UE $k$. 
The received signal $\mathbf{Y}_l^p \in \mathbb{C}^{N \times \tau_p}$ at AP $l$ is
\begin{equation}
\mathbf{Y}_l^p = \sum_{i=1}^{K} \sqrt{p_i}\mathbf{h}_{il}\boldsymbol{\phi}_{t_i}^T + \mathbf{N}_l,
\end{equation}
where $p_i \geq 0$ is the transmit power of UE $i$ and $\mathbf{N}_l \in \mathbb{C}^{N \times \tau_p}$ is the noise at the receiver modeled with independent entries distributed as $\mathcal{CN}\left(0,\sigma^2\right)$ with $\sigma^2$ being the noise power. The estimation of $\mathbf{h}_{kl}$ at AP $l$ proceeds in two phases: first despreading of the received signal is done and then the MMSE estimator is employed. Accordingly, the MMSE channel estimate $\widehat{\mathbf{h}}_{kl}\in \mathbb{C}^{N\times1}$ is given by \cite{stevenkay}
\begin{equation}
\begin{aligned}
\widehat{\mathbf{h}}_{kl}=\sqrt{p_k \tau_p}\mathbf{R}_{kl}\boldsymbol{\Psi}_{t_k l}^{-1}\mathbf{y}_{t_k l}^p,
\end{aligned}
\end{equation}
where
\begin{align}
\mathbf{y}_{t_k l}^p & = \mathbf{Y}_l^p\frac{\boldsymbol{\phi}_{t_k}^{*}}{\sqrt{\tau_p}}\notag\\
&=  \sum_{i \in \mathcal{S}_k}\sqrt{p_i\tau_p}\mathbf{h}_{il} + \mathbf{n}_{t_k l} ,
\end{align}
\begin{align}
\boldsymbol{\Psi}_{t_k l} & = \mathbb{E}\left\{\left(\mathbf{y}_{t_k l}^p-\mathbb{E}\{\mathbf{y}_{t_k l}^p\}\right)\left(\mathbf{y}_{t_k l}^p-\mathbb{E}\{\mathbf{y}_{t_k l}^p\}\right)^H\right\}\notag\\ & = \sum_{i \in \mathcal{S}_k} \tau_p p_i\mathbf{R}_{il} + \sigma^2 \mathbf{I}_N
\end{align}
is the despreaded signal and its covariance matrix, respectively. Here, $\mathbf{n}_{t_k l} \triangleq \mathbf{N}_l {\boldsymbol{\phi}_{t_k}^*}/{\sqrt{\tau_p}} \sim \mathcal{CN}\left(\mathbf{0},\sigma^2\mathbf{I}_{N}\right)$ is the effective noise. An important consequence of MMSE estimation is the statistical independence of the estimate $\widehat{\mathbf{h}}_{kl} \sim \mathcal{CN}(\mathbf{0},\widehat{\mathbf{R}}_{kl})$ and the estimation error $\widetilde{\mathbf{h}}_{kl} = \mathbf{h}_{kl} -\widehat{\mathbf{h}}_{kl} \sim \mathcal{CN}(\mathbf{0},\widetilde{\mathbf{R}}_{kl})$ with
\begin{equation}
\begin{aligned}[b]
\widehat{\mathbf{R}}_{kl} & = \mathbb{E}\left\{\left(\widehat{\mathbf{h}}_{kl}-\mathbb{E}\{\widehat{\mathbf{h}}_{kl}\}\right)\left(\widehat{\mathbf{h}}_{kl}-\mathbb{E}\{\widehat{\mathbf{h}}_{kl}\}\right)^H\right\}\\
&=p_k\tau_p\mathbf{R}_{kl}\boldsymbol{\Psi}_{t_k l}^{-1}\mathbf{R}_{kl},
\end{aligned}
\end{equation}
\begin{equation}
\begin{aligned}[b]
\widetilde{\mathbf{R}}_{kl} & = \mathbb{E}\left\{\left(\widetilde{\mathbf{h}}_{kl}-\mathbb{E}\{\widetilde{\mathbf{h}}_{kl}\}\right)\left(\widetilde{\mathbf{h}}_{kl}-\mathbb{E}\{\widetilde{\mathbf{h}}_{kl}\}\right)^H\right\}\\ & = \mathbf{R}_{kl} -\widehat{\mathbf{R}}_{kl}
\end{aligned}
\end{equation}
as the respective covariance matrices. 

\subsection{Uplink Payload Transmission}
During the uplink payload transmission, the received signal $\mathbf{y}_l \in \mathbb{C}^N$ at AP $l$ is given by 
\begin{equation}\label{ULrecSig2}
\mathbf{y}_l = \mathbf{H}_{l}\mathbf{s} + \mathbf{n}_l,
\end{equation}
where $\mathbf{H}_l = [\mathbf{h}_{1l},\mathbf{h}_{2l},\cdots,\mathbf{h}_{Kl}]\in \mathbb{C}^{N\times K}$ is the channel matrix, $\mathbf{s} = [s_{1},s_{2},\cdots,s_{K}]^T \in \mathbb{C}^{K}$ is the signal vector with $s_k \sim \mathcal{CN}\left(0,p_k\right)$ being the payload signal transmitted by UE $k$ with power $p_k$ and the receiver noise vector $\mathbf{n}_l \sim \mathcal{CN}\left(\mathbf{0},\sigma^2 \mathbf{I}_N\right)$. We assume $s_k$ is independent of $s_m$ for $\ k \neq m$ and the signal vector is distributed as $\mathbf{s}\sim \mathcal{CN}\left(\mathbf{0},\mathbf{Q}\right)$ with $\mathbf{Q} = \textrm{diag}(p_1,\cdots,p_K)$.  Let $\mathbf{H}_{l} = \widehat{\mathbf{H}}_l + \widetilde{\mathbf{H}}_l$ with $ \widehat{\mathbf{H}}_l = [\widehat{\mathbf{h}}_{1l},\widehat{\mathbf{h}}_{2l},\cdots\widehat{\mathbf{h}}_{Kl}]$ being the matrix with channel estimates and $ \widetilde{\mathbf{H}}_l = [\widetilde{\mathbf{h}}_{1l},\widetilde{\mathbf{h}}_{2l},\cdots\widetilde{\mathbf{h}}_{Kl}]$ the  matrix with estimation errors. Accordingly, \eqref{ULrecSig2} is equivalent to
\begin{equation}
\begin{aligned}
\mathbf{y}_l &= \widehat{\mathbf{H}}_l\mathbf{s} + \widetilde{\mathbf{H}}_l\mathbf{s} + \mathbf{n}_l\\ &=  \widehat{\mathbf{H}}_l\mathbf{s} + \mathbf{w}_l,
\end{aligned}
\end{equation}
where $\mathbf{w}_l = \widetilde{\mathbf{H}}_l\mathbf{s} + \mathbf{n}_l$ can be thought of as a colored noise vector with zero mean and covariance matrix
\begin{equation}
\mathbf{\Sigma}_l = \sum_{i=1}^{K}p_i\widetilde{\mathbf{R}}_{il} + \sigma^2\mathbf{I}_N.
\end{equation}
Note that $\mathbf{s}$ is uncorrelated with $\mathbf{w}_l$ but is statistically dependent. 
We denote the estimate of $\mathbf{s}$ at AP $l$ as $\widehat{\mathbf{s}}_l = [\widehat{s}_{1l},\widehat{s}_{2l},\cdots,\widehat{s}_{Kl}]^T$ with $\widehat{s}_{kl}$ being the estimate of $s_k$. We will analyze different ways to compute the estimate in the next section.
\section{Sequential Uplink Processing}\label{SeqProc}
In this section, we derive an optimal sequential receiver algorithm in the sense of simultaneously achieving the maximum SE and the minimum MSE at the CPU. We first briefly describe the optimal receiver for the centralized cell-free mMIMO network \cite{NeyebiMMSE,centBjorn} which is a baseline for performance. We introduce following notation which is utilized throughout this paper:
\begin{equation}\label{augReceiveSig}
\mathbf{z}_l = \widehat{\mathbf{G}}_l\mathbf{s}+\overline{\mathbf{w}}_l
\end{equation}
where 
\begin{equation}\label{augtNotation}
\begin{aligned}
\mathbf{z}_l &= \left[\mathbf{y}_1^H,\cdots,\mathbf{y}_l^H\right]^H,\\
\widehat{\mathbf{G}}_l &=  \left[\widehat{\mathbf{H}}_1^H,\cdots,\widehat{\mathbf{H}}_l^H\right]^H,\\
\overline{\mathbf{w}}_l\ &= \left[\mathbf{w}_1^H,\cdots,\mathbf{w}_l^H\right]^H,
\end{aligned}
\end{equation}
where $\mathbf{z}_l$ can be thought of as the augmented received signal at AP $l \in \{1,\cdots,L\}$, $\widehat{\mathbf{G}}_l$ contains the augmented channel estimates and $\overline{\mathbf{w}}_l$ being the augmented colored noise vector with zero mean and covariance matrix $
\mathbf{K}_l= \textrm{diag}\left(\mathbf{\Sigma}_1,\cdots,\mathbf{\Sigma}_l\right)$.

\subsection{Optimal Centralized Implementation}\label{centOptimal}
In a centralized setup, all the $L$ APs send the received pilot and payload signals to the CPU which after computing channel estimates forms the following augmented received signal:
\begin{equation}\label{centReceiveSig}
\mathbf{z}_L = \widehat{\mathbf{G}}_L\mathbf{s}+\overline{\mathbf{w}}_L.
\end{equation}
The CPU employs the LMMSE receiver which is optimal in both the maximum SE and the minimum MSE sense to compute the data signal estimate as \cite{centBjorn}
\begin{equation}\label{centEst}
\widehat{\mathbf{s}}_L^{c} = \mathbf{V}_L^{c}\mathbf{z}_L,
\end{equation}
where $\widehat{\mathbf{s}}_L^{c} $ is the signal estimate, $\mathbf{V}_L^{c}$ is the receive combining matrix and "c" indicates that it is a centralized scheme. The LMMSE receiver matrix is given as \cite{stevenkay}
\begin{equation}\label{centLMMSE}
\mathbf{V}_L^\textrm{c} = \mathbf{Q}\widehat{\mathbf{G}}_L^H \mathbf{\Lambda}_L^{-1},
\end{equation}
where
\begin{equation}
\mathbf{\Lambda}_L = \left(\mathbf{K}_L + \widehat{\mathbf{G}}_L\mathbf{Q}\widehat{\mathbf{G}}_L^H\right).
\end{equation}
The achievable SE of UE $k$ using the LMMSE receiver is given by \cite{centBjorn}
\begin{equation}\label{seCPU}
\textup{SE}_k^c = \left(1-\frac{\tau_p}{\tau_c}\right)\mathbb{E}\left\{\textup{log}_2\left(1 + \Gamma_k^{c}\right)\right\},
\end{equation}
where the instantaneous effective signal-to-interference-and-noise ratio (SINR) is
\begin{equation}
\Gamma_k^{c} = p_k \widehat{\mathbf{h}}_k^H \left(\sum_{i=1,i\neq k}^{K}p_i\widehat{\mathbf{h}}_i\widehat{\mathbf{h}}_i^H + \mathbf{K}_L\right)^{-1}\widehat{\mathbf{h}}_k.
\end{equation}
The minimum MSE at the CPU is
\begin{equation}\label{mseCPU}
\begin{aligned}
\mathbf{P}_L^c  &=\mathbb{E}\left\{\left(\mathbf{s}-\widehat{\mathbf{s}}_L^{c}\right)\left(\mathbf{s}-\widehat{\mathbf{s}}_L^{c}\right)^H\ |\ \widehat{\mathbf{G}}_L\right\}\\
&= \mathbf{Q} - \mathbf{V}_L^{c} 
\widehat{\mathbf{G}}_{L}\mathbf{Q}.
\end{aligned}
\end{equation}

\subsection{Sequential Linear Processing}
In this subsection, we present a generic class of linear processing algorithms suitable for sequential cell-free mMIMO networks. In general, a sequential processing algorithm starts with AP 1 computing the soft estimate $\widehat{\mathbf{s}}_1$ of the signal $\mathbf{s}$ and forwards the computed estimate and also some useful side information to AP 2 which would help it in making a better estimate. Now, AP $l\in \{2,\cdots,L\}$ upon receiving the information from AP $(l-1)$ computes the soft estimate $\widehat{\mathbf{s}}_l$ of the signal $\mathbf{s}$ and then it forwards the computed estimate along with some useful side information to AP $(l+1)$. This sequential process continues till AP $L$ which forwards the final estimate
$\widehat{\mathbf{s}}_L$ of the signal to the CPU. The CPU might also be co-located with AP $L$.

Beginning with AP 1, the soft estimate $\widehat{\mathbf{s}}_1$ of $\mathbf{s}$ computed by AP 1 is given by
\begin{equation}\label{genEstAP1}
\widehat{\mathbf{s}}_1 = \mathbf{B}_{1}\mathbf{y}_1,
\end{equation}
where $\mathbf{B}_1 \in \mathbb{C}^{K\times N}$ is any receiver combining matrix that can be selected based on the information available at AP $1$ (e.g., channel estimate matrix, channel and noise statistics). AP 1 then forwards the estimate \eqref{genEstAP1} along with other useful side information to AP 2. Then, AP $l\in \{2,\ldots,L\}$ upon receiving the estimate $\widehat{\mathbf{s}}_{(l-1)}$ and other side information from AP $(l-1)$ computes its estimate as follows:
\begin{equation}\label{genSeq}
\widehat{\mathbf{s}}_l = \mathbf{A}_l\widehat{\mathbf{s}}_{(l-1)} + \mathbf{B}_l\mathbf{y}_l,\ \ l\in\{2,\cdots,L\}
\end{equation}
where $\mathbf{A}_l \in \mathbb{C}^{K\times K}$ and $\mathbf{B}_l \in \mathbb{C}^{K\times N}$ are some receiver combining matrices which depend on the information available at AP $l$. Considering the initial estimate $\widehat{\mathbf{s}}_0$ as the zero vector, the computation in \eqref{genSeq} can be generalized to any AP $l \in \{1,\ldots,L\}$. We now present the achievable SE and the MSE at the CPU for any generic sequential linear processing of the form \eqref{genSeq}. First note that, \eqref{genSeq} can be equivalently written as
\begin{equation}\label{genSeq_Cent}
\widehat{\mathbf{s}}_l = \overline{\mathbf{B}}_l\mathbf{z}_l
\end{equation}
where,
\begin{equation}\label{Bbar}
\overline{\mathbf{B}}_{l}= \begin{cases}
\begin{bmatrix}
\mathbf{A}_l\overline{\mathbf{B}}_{(l-1)}\ \ \mathbf{B}_l
\end{bmatrix}, & l> 1\\
\mathbf{B}_1, & l= 1.
\end{cases}
\end{equation}
The following two propositions provide closed-form expressions for the achievable SE and the MSE at the CPU, respectively. Let the augmented receiver matrix in \eqref{Bbar} at the AP $L$ be
\begin{equation}\label{Bbar2}
\overline{\mathbf{B}}_{L} = [\mathbf{b}_1,\cdots,\mathbf{b}_K]^H
\end{equation}
with $\mathbf{b}_k$ being the combining vector for $k$th UE and $\widehat{\mathbf{h}}_k$ be the $k$th column of channel estimation matrix $\widehat{\mathbf{G}}_L$ or equivalently $\widehat{\mathbf{h}}_k = [\widehat{\mathbf{h}}_{k1}^H,\cdots,\widehat{\mathbf{h}}_{kL}^H]^H$ at the final AP $L$. Thus, note the $k$th UE estimate $\widehat{s}_{kL}$ at the CPU is a function of $\mathbf{b}_k$ i.e., $\widehat{s}_{kL}(\mathbf{b}_k)$.
\begin{prop}\label{propSE}
The achievable SE of UE $k$ for any receiver algorithm of the form \eqref{genSeq} is
\begin{equation}\label{SEgen}
\textup{SE}_k (\mathbf{b}_k) = \left(1-\frac{\tau_p}{\tau_c}\right)\mathbb{E}\left\{\textup{log}_2\left(1 + \Gamma_k^{'}(\mathbf{b}_k)\right)\right\},
\end{equation}
where the instantaneous effective SINR is
\begin{equation}
\Gamma_k^{'}\left(\mathbf{b}_k\right) = \frac{p_k\vert \mathbf{b}_k^H\widehat{\mathbf{h}}_k\vert^2}{\sum_{i=1,i\neq k}^{K}p_i\vert \mathbf{b}_k^H\widehat{\mathbf{h}}_i\vert^2+ \mathbf{b}_k^H\mathbf{K}_L\mathbf{b}_k}.
\end{equation}
\end{prop}
\begin{proof}
	The proof follows from Proposition 1 in \cite{centBjorn}.
\end{proof}
\begin{prop}\label{propMSE}
The MSE of UE $k$ at the CPU for any receiver algorithm of the form \eqref{genSeq} for the given side information including channel estimates is
\begin{equation} \label{MSEgen}
\begin{aligned}
e_k^{'}(\mathbf{b}_k) &= \min_{\mathbf{b}_k} \mathbb{E}\{\vert s_k - \widehat{s}_{kL}(\mathbf{b}_k)\vert^2|\textup{side information}\}\\
&= p_k - 2\Re\{\mathbf{b}_k^H\widehat{\mathbf{h}}_k\}p_k + \mathbf{b}_k^H\mathbf{\Lambda}_L\mathbf{b}_k.\\
\end{aligned}
\end{equation}
\end{prop}
One special case of the sequential process in \eqref{genSeq} is MR, which is obtained when the receiver matrix $\mathbf{A}_l$ is a identity matrix and 
\begin{equation}\label{SMR_1}
\mathbf{B}_l =[\widehat{\mathbf{h}}_{1l},\cdots,\widehat{\mathbf{h}}_{Kl}]^H.
\end{equation}
One obtains the following estimate at AP $L$:
\begin{equation}\label{SMR_2}
\widehat{s}_{kL} = \sum_{l=1}^{L}\widehat{\mathbf{h}}_{kl}^H\mathbf{y}_l = \widehat{\mathbf{h}}_k^H \mathbf{z}_L.
\end{equation}

In the next subsection, we provide an optimal choice of receiver combining matrices $\{\mathbf{A}_l\}$ and $\{\mathbf
B_l\}$ among the class of generic receivers which not only maximizes \eqref{SEgen} but also minimizes \eqref{MSEgen}. 
\subsection{Optimal Sequential Linear Processing (OSLP)} 
We will now present a particular sequential algorithm of the form \eqref{genSeq} which we will show that it achieves the same optimal performance as centralized LMMSE but with lower fronthaul requirements.

In the proposed algorithm, the receiver combining matrices are chosen as follows:
\begin{align}
\mathbf{A}_l &= \mathbf{I}_K-\mathbf{T}_{l}\widehat{\mathbf{H}}_l\label{oslpAl}\\
\mathbf{B}_l &=  \mathbf{T}_{l}\label{oslpBl},
\end{align}
where
\begin{equation}\label{Tl}
\mathbf{T}_{l} = \mathbf{P}_{(l-1)}\widehat{\mathbf{H}}_l^H\left(\mathbf{\Sigma}_l + \widehat{\mathbf{H}}_l\mathbf{P}_{(l-1)}\widehat{\mathbf{H}}_l^H\right)^{-1}
\end{equation}
and
\begin{equation}
\begin{aligned}
\mathbf{P}_{(l-1)} = \left(\mathbf{I}_K - \mathbf{T}_{(l-1)}\widehat{\mathbf{H}}_{(l-1)}\right)\mathbf{P}_{l-2}.
\end{aligned}
\end{equation}
Thus, substituting \eqref{oslpAl} and \eqref{oslpBl} in \eqref{genSeq} gives soft estimate at AP $l$ as
\begin{equation}\label{Eq_shatL1}
\widehat{\mathbf{s}}_l = \widehat{\mathbf{s}}_{(l-1)} + \mathbf{T}_l\left(\mathbf{y}_l - \widehat{\mathbf{H}}_l \widehat{\mathbf{s}}_{(l-1)}\right).
\end{equation}
We will prove later that $\mathbf{T}_l$ is an optimal local LMMSE receiver i.e., it minimizes
\begin{equation}
\mathbb{E}\left\{\Vert \mathbf{s} - \widehat{\mathbf{s}}_l \Vert^2\ |\ \widehat{\mathbf{H}}_l,\ \widehat{\mathbf{s}}_{(l-1)},\mathbf{P}_{(l-1)}\right\}
\end{equation}
and $\mathbf{P}_{l}$ is the error covariance matrix at AP $l$ i.e.,
\begin{equation}
\mathbf{P}_{l} = \mathbb{E}\left\{\left(\mathbf{s}-\widehat{\mathbf{s}}_{l}\right)\left(\mathbf{s}-\widehat{\mathbf{s}}_{l}\right)^H\ |\ \widehat{\mathbf{H}}_l,\ \widehat{\mathbf{s}}_{(l-1)},\mathbf{P}_{(l-1)}\right\}.
\end{equation}
After AP $L$ computes the final estimate, it forwards the estimate $\widehat{\mathbf{s}}_L$ to the CPU where the final decoding of the signal is done. The above described algorithm is presented in the form of a pseudo-code in Algorithm \ref{Algo1}. Observe that, each AP $l\in\{1,\cdots,L\}$ can simultaneously compute $\mathbf{T}_{l}\widehat{\mathbf{H}}_{l}$ and $\left(\mathbf{I}_K - \mathbf{T}_{l}\widehat{\mathbf{H}}_{l}\right)$ once for $\tau_c-\tau_p$ channel uses. So in each $\tau_c-\tau_p$ channel use, every AP just has to perform multiplication and addition of signals and forward to the consecutive AP.
\begin{algorithm}[t]
\caption{Optimal Sequential Linear Processing (OSLP) for Radio Stripe}
\begin{algorithmic} \label{Algo1}
	\STATE  1. \textbf{Initialize}: $\widehat{\mathbf{s}}=\mathbf{0},\ \mathbf{P}_0 = \mathbf{Q}$;
	\STATE 2. \textbf{for} $l = 1:L$
	\STATE \quad\ (i)  Compute $\mathbf{T}_{l} = \mathbf{P}_{(l-1)}\widehat{\mathbf{H}}_l^H\left(\mathbf{\Sigma}_l + \widehat{\mathbf{H}}_l\mathbf{P}_{(l-1)}\widehat{\mathbf{H}}_l^H\right)^{-1}$ 
	\STATE \quad\ (ii) Compute $\widehat{\mathbf{s}}_l = \widehat{\mathbf{s}}_{(l-1)} + \mathbf{T}_l(\mathbf{y}_l - \widehat{\mathbf{H}}_l \widehat{\mathbf{s}}_{(l-1)})$
	\STATE \quad\ (iii) Compute $\mathbf{P}_{l} = \left(\mathbf{I}_K - \mathbf{T}_{l}\widehat{\mathbf{H}}_{l}\right)\mathbf{P}_{(l-1)}$
	\STATE \quad \textbf{end}
	\STATE 3. \textbf{Output}: $\widehat{\mathbf{s}}_L$ and $\mathbf{P}_L$
\end{algorithmic}
\end{algorithm}

We will now show that the proposed algorithm is optimal in the sense of minimizing the MSE (i.e., showing that $\mathbf{P}_L = \mathbf{P}_L^c$) and simultaneously maximizing the SE (i.e., it has the same spectral efficiency as in \eqref{seCPU}).
Note that \eqref{Eq_shatL1} can be re-written as in \eqref{genSeq_Cent}
\begin{equation}\label{kalman2cent}
\begin{aligned}
\widehat{\mathbf{s}}_l &=\begin{bmatrix}
\left(\mathbf{I}-\mathbf{T}_{l}\widehat{\mathbf{H}}_l\right)\ \mathbf{T}_{l}
\end{bmatrix}\begin{bmatrix}
\widehat{\mathbf{s}}_{(l-1)}\\ \mathbf{y}_l
\end{bmatrix}\\
&=\begin{bmatrix}
(\overline{\mathbf{V}}_{(l-1)}-\mathbf{T}_l\widehat{\mathbf{H}}_l\overline{\mathbf{V}}_{(l-1)})\ \ \mathbf{T}_l
\end{bmatrix}\mathbf{z}_l\\
&=\overline{\mathbf{V}}_{l}\mathbf{z}_l,
\end{aligned}
\end{equation}
where
\begin{equation}\label{vbarL}
\overline{\mathbf{V}}_{l}= \begin{cases}
\begin{bmatrix}
(\overline{\mathbf{V}}_{(l-1)}-\mathbf{T}_l\widehat{\mathbf{H}}_l\overline{\mathbf{V}}_{(l-1)})\ \ \mathbf{T}_l
\end{bmatrix}, & l> 1\\
\mathbf{T}_1, & l= 1.
\end{cases}
\end{equation}
Equation \eqref{kalman2cent} is important since it establishes the relationship between the proposed sequential processing estimate and the centralized processing estimate in \eqref{centEst}. We prove the proposed algorithm is optimal with the help of following theorem.
\begin{theorem}\label{TheoremOSLP}
In the sequential processing with Algorithm \ref{Algo1}, the estimate obtained at AP $L$ is equivalent to that obtained by centralized processing with LMMSE receiver i.e.,
\begin{equation}\label{shat_oslp_cent}
\widehat{\mathbf{s}}_L = \widehat{\mathbf{s}}_L^{c},
\end{equation}
\end{theorem}
\begin{proof}
The proof is given in Appendix A.
\end{proof}

In a nutshell, Theorem \ref{TheoremOSLP} shows that it is possible to decentralize the LMMSE receiver to a sequential implementation using the Algorithm \ref{Algo1} and we will call this algorithm as optimal sequential linear processing (OSLP) since it is optimal in the generic class of sequential linear processing in the sense of the SE. We introduce the following notations to use in the corollaries that follows:
\begin{align}
\mathbf{A}_l^{o} &= \mathbf{I}_K-\mathbf{T}_{l}\widehat{\mathbf{H}}_l\label{oslpAo}\\
\mathbf{B}_l^{o} &=  \mathbf{T}_{l}\label{oslpBo},
\end{align}
be the OSLP receiver matrices at AP $L$ from \eqref{oslpAl} and \eqref{oslpBl}. The important consequences of Theorem \ref{TheoremOSLP} are presented as corollaries below:
\begin{corollary}
Comparing \eqref{kalman2cent} and \eqref{genSeq_Cent} and from Theorem \ref{TheoremOSLP}, the following relation holds:
\begin{equation}
\textup{SE}_k\left(\{\mathbf{A}_l,\mathbf{B}_l\}\right) \leq \textup{SE}_k(\mathbf{V}_L^c),\ l=1,\ldots,L
\end{equation}
where $\textup{SE}_k(\cdot)$ is the achievable SE of $k$th UE, $\{\mathbf{A}_l,\mathbf{B}_l\}$ are receiver matrices for any generic sequential linear processing given in \eqref{genSeq} and $\mathbf{V}_L^c$ is the centralized LMMSE receiver given in \eqref{centLMMSE}. Equality is achieved with the proposed OSLP algorithm i.e., when $\{\mathbf{A}_l,\mathbf{B}_l\}=\{\mathbf{A}_l^o,\mathbf{B}_l^o\}$. A rigorous lower bound on the capacity using the OSLP can be obtained in a closed form by plugging $\overline{\mathbf{B}}_{L}=\overline{\mathbf{V}}_{L}$ in \eqref{Bbar} and using Proposition \ref{propSE}. 
The maximum achievable SE of UE $k$ is given by
\begin{equation}\label{SE_OSLP}
\textup{SE}_k = \left(1-\frac{\tau_p}{\tau_c}\right)\mathbb{E}\left\{\textup{log}_2(1 + \Gamma_k^{\textrm{max}})\right\},
\end{equation}
where $\Gamma_k^{\textrm{max}}$ is the maximum instantaneous effective SINR given by
\begin{equation}\label{maxSNR}
\Gamma_k^{\textrm{max}} = p_k \widehat{\mathbf{h}}_k^H \left(\sum_{i=1,i\neq k}^{K}p_i\widehat{\mathbf{h}}_i\widehat{\mathbf{h}}_i^H + \mathbf{K}_L\right)^{-1}\widehat{\mathbf{h}}_k.
\end{equation}
\end{corollary}
\begin{corollary}
Comparing \eqref{kalman2cent} and \eqref{genSeq_Cent} and from Theorem \ref{TheoremOSLP}, the following relation holds:
\begin{equation}
e_k^{'}(\mathbf{V}_L^c) \leq  e_k^{'}\left(\{\mathbf{A}_l,\mathbf{B}_l\}\right),\ l=1,\ldots,L
\end{equation}
where $e_k^{'}(\cdot)$ is the MSE of the $k$th UE at the CPU. Equality is achieved with the proposed OSLP algorithm i.e., when $\{\mathbf{A}_l,\mathbf{B}_l\}=\{\mathbf{A}_l^o,\mathbf{B}_l^o\}$. The proposed OLSP algorithm achieves the minimum MSE $e_k^{\textrm{min}}$ given the side information for UE $k$ at the CPU which is computed in closed form by taking the $k$th diagonal entry of the error covariance matrix $\mathbf{P}_L^c$ (since $\mathbf{P}_L=\mathbf{P}_L^c$ ) in \eqref{mseCPU} and is given by
\begin{equation}\label{minMSE}
e_k^{\textrm{min}} = p_k - p_k^2\widehat{\mathbf{h}}_k^H \mathbf{\Lambda}_L^{-1}\widehat{\mathbf{h}}_k.
\end{equation}
\end{corollary}
\begin{corollary}\label{corrMonotonic}
The achievable SINR $\Gamma_k$ of UE $k$ increases monotonically with an increase in the number of APs $L$. On other hand, the MSE at the CPU, $e_k$ being inversely related to $\Gamma_k$ as
\begin{equation}\label{mse_se}
e_k^{\textrm{min}} =  \frac{p_k}{1+\Gamma_k^{\textrm{max}}},
\end{equation} 
decreases monotonically with the increase in $L$. 
\end{corollary}
\begin{proof}
For an optimal centralized scheme, with increase in number of APs, say AP $L$ to AP $(L+1)$, the available information at the CPU increases from $\{\mathbf{y}_1,\cdots,\mathbf{y}_{L}\}$ to $\{\mathbf{y}_1,\cdots,\mathbf{y}_{L+1}\}$. As there is no loss in the original information i.e., $\{\mathbf{y}_1,\cdots,\mathbf{y}_L\}$, it implies that the SINR (or the MSE) also monotonically increases (or decreases). From Theorem \ref{TheoremOSLP}, the same result follows for the proposed OSLP algorithm with increase in APs.
\end{proof}
\begin{corollary}\label{corrorderingInv}
The performance of the OSLP algorithm is invariant to the order in which the sequential processing is implemented i.e. the ordering of the APs.
\end{corollary}
\begin{proof}
In a centralized processing with LMMSE receiver matrix, the estimate of the signal is obtained by utilizing the complete information received from all the APs in a cell-free setup. Hence, the performance will have no effect if the order in which the signals are received at the CPU is changed. Since the OSLP performance is equal to the centralized LMMSE processing from Theorem \ref{TheoremOSLP}, the required result follows.
\end{proof}
It is worth noting that the Corollary \ref{corrorderingInv} is not true in general for any generic sequential linear processing in \eqref{genSeq}. For instance, it doesn't hold for the algorithms in \cite{icc2020} and \cite{rls2}.

There are two important practical benefits of the OSLP algorithm. The first benefit is that it makes use of local processing capabilities at the APs instead of requiring a CPU with a fast processor. The second benefit is that the achievable SE (or the MSE at the	 CPU) monotonically increase (or MSE decreases) while maintaining the same fronthaul capacity requirements in each link between the APs. We provide mathematical update equations for the SE and the MSE in the following proposition:
\begin{prop}\label{propUpdateEqs}
Let $e_{kl}$ and $\Gamma_{kl}$ denote the MSE and SINR of UE k achieved at AP l when using the OSLP algorithm for given channel estimates. The MSEs achieved at the adjacent APs can be computed using
\begin{equation}
e_{kl} = e_{k(l-1)} - \alpha_{kl},
\end{equation}
where
\begin{equation}
\alpha_{kl} = \left[\mathbf{T}_l \widehat{\mathbf{H}}_l \mathbf{P}_{(l-1)}\right]_{kk}.
\end{equation}
Note that the matrix $\mathbf{T}_l \widehat{\mathbf{H}}_l \mathbf{P}_{(l-1)}$ is non-negative definite and hence $\alpha_{kl}\geq 0$. Using \eqref{mse_se}, the update equation for the achievable SINR, $\Gamma_{kl}$ of UE $k$ at AP $l$ is
\begin{equation}
\begin{aligned}
\Gamma_{kl} = \Gamma_{k(l-1)} +  \gamma_{kl},
\end{aligned}
\end{equation}
where 
\begin{equation}
\gamma_{kl} = \frac{\alpha_{kl}\left(\Gamma_{k(l-1)}+1\right)^2}{p_k - \alpha_{kl}\left(\Gamma_{k(l-1)}+1\right)}.
\end{equation}
Finally, using the logarithmic equality, $\textup{log}_2(a+b)=\textup{log}_2(a)+\textup{log}_2(1+\frac{b}{a})$ and letting $a = 1+\Gamma_{k(l-1)}, b = \gamma_{kl} $, we obtain the update equation of SE of UE $k$ at AP $l$ as
\begin{equation}\label{updateEqSE}
\textup{SE}_{kl} = \textup{SE}_{k(l-1)} + \zeta_{kl}
\end{equation}
where $\zeta_{kl} = \mathbb{E}\left\{\textup{log}_2(1+\frac{\gamma_{kl}}{1+\Gamma_{k(l-1)}})\right\}$.
\end{prop}
\subsection{Normalized LMMSE based sequential processing}\label{NLMMSE}
In the conference version of this paper \cite{icc2020}, a sequential processing algorithm using normalized LMMSE (N-LMMSE)  scheme was presented which is given as a pseudo-code in Algorithm \ref{Algo3}. The main difference between the algorithm presented in Algorithm \ref{Algo3} and the OSLP algorithm is that the former computes the $k$th UE's estimate $\widehat{s}_{kl}$ using only $\widehat{s}_{k(l-1)}$, whereas OSLP uses all the signal estimates in $\widehat{\mathbf{s}}_{(l-1)}$. Hence, the computation of the signal estimate in the OSLP algorithm makes use of more available information than Algorithm \ref{Algo3}, thus the OSLP algorithm always performs superior or equal to Algorithm \ref{Algo3} in the sense of minimum MSE and also maximum SE.

For Algorithm \ref{Algo3}, besides being suboptimal to the OSLP algorithm, the fronthaul signaling in each link between the APs is actually higher. This is quantitatively explained in detail in Section \ref{fronthaul}. 
\begin{algorithm}[t]
\caption{Sequential N-LMMSE Processing from \cite{icc2020}}
\begin{algorithmic} \label{Algo3}
	\STATE  1. Compute local LMMSE receiver, $\mathbf{V}_{1}= [\mathbf{v}_{11},\cdots,\mathbf{v}_{1K}]^H$ given
	$\widehat{\mathbf{H}}_1$ and $\mathbf{\Sigma}_1$
	\STATE 2. Compute $\widehat{\mathbf{s}}_1 = \mathbf{V}_{1}\mathbf{y}_1$ and initialize: $\widehat{\mathbf{H}}_{k1}^{'}=\widehat{\mathbf{H}}_1,\ \mathbf{\Sigma}_{kl}^{'} = \mathbf{\Sigma}_1,\ \forall k \in \{1,\cdots,K\}$
	\STATE 4. \textbf{for} $l=2:L$
	\STATE \quad \quad \textbf{for} $k=1: K$
	\STATE \qquad (a) Compute LMMSE receiver $\mathbf{v}_{kl}$ for $k$th UE, given $\widehat{\mathbf{H}}_{kl}^{'}$ and $ \mathbf{\Sigma}_{kl}^{'}$ where,
	\STATE \qquad \qquad $\widehat{\mathbf{H}}_{kl}^{'} = [\widehat{\mathbf{H}}_{k(l-1)}^{'H}\mathbf{v}_{k(l-1)}, \widehat{\mathbf{H}}_l^H]^H$ is augmented channel estimate and 
	\STATE \qquad  \qquad$\mathbf{\Sigma}_{kl}^{'} = \textrm{diag}(\mathbf{v}_{k(l-1)}^H\mathbf{\Sigma}_{k(l-1)}^{'}\mathbf{v}_{k(l-1)},\mathbf{\Sigma}_l)$ is augmented noise covariance matrix
	\STATE \qquad	 (b) Compute $\widehat{s}_{kl} = \mathbf{v}_{kl}^H[\widehat{s}_{kl}^{~*},\mathbf{y}_l^H]^H $
	\STATE \qquad \textbf{end}
	\STATE \quad \textbf{end}
	\STATE 5. \textbf{Output}: $\widehat{\mathbf{s}}_L$
\end{algorithmic}
\end{algorithm}
\section{Fronthaul Signalling and Latency}\label{fronthaul}
In this section, we analyze two practical aspects crucial for radio stripes implementation, namely the fronthaul signaling capacity in each link connecting the APs and the latency in the network. We study these factors for the different algorithms described in Section \ref{SeqProc}. We define the fronthaul signaling quantitatively as the total number of real symbols that each AP shares to its consecutive AP in an arbitrary coherence block. 

\subsection{Centralized LMMSE Implementation}
In the fully centralized processing scheme the signaling increases along the fronthaul since the signals from every AP much reach the CPU without being merged with the signals from other APs. Hence, we consider the fronthaul signaling to be the total number of real symbols that are being transmitted in all the links between the APs and the CPU, which is the capacity that is required between AP L and the CPU. We assume that the CPU has the information of channel spatial statistics i.e., $\mathbf{R}_{kl}$. Each AP forwards $\tau_cN$ complex symbols corresponding to the received pilot and data signals to the CPU. This amounts to $\tau_cNL$ complex symbols or equivalently $2\tau_cNL$ real symbols being transmitted from the APs to the CPU.
\subsection{The OSLP Algorithm}
In each link between the APs, the following amount of information is being shared:
\begin{enumerate}[(i)]
\item Each AP forwards the computed signal estimate $\widehat{\mathbf{s}}_l$ of the signal $\mathbf{s}$ to the successive AP i.e., from AP $l$ to $(l+1),\ l =1,\cdots L-1$. This corresponds to $2K(\tau_c - \tau_p)$ real symbols.
\item The error covariance matrix $\mathbf{P}_k$, which corresponds to $K^2$ real symbols.
\end{enumerate}
\subsection{The N-LMMSE based sequential processing from \cite{icc2020} (Algorithm \ref{Algo3})}
As described briefly in Section \ref{NLMMSE}, the amount of data that is being shared in each link between the APs is:
\begin{enumerate}[(i)]
\item Each AP forwards the computed signal estimate $\widehat{\mathbf{s}}_l$ of the signal $\mathbf{s}$ to the next AP, thus $2K(\tau_c - \tau_p)$ real symbols.
\item Effective channel estimates $\widehat{\mathbf{H}}_{(l-1)}^{'H}\mathbf{v}_{k(l-1)}\ \forall k$ that is represented by $2K^2$ real symbols.
\item Effective channel matrix estimation statistics $\mathbf{v}_{k(l-1)}^H\mathbf{\Sigma}_{(l-1)}^{'}\mathbf{v}_{k(l-1)}\ \forall k$ which amounts to $K$ real symbols.
\end{enumerate}
The total fronthaul signaling required is $2K^2 + K+ 2K(\tau_c - \tau_p)\ \textrm{real symbols}$. The fronthaul mentioned in \cite{icc2020} is $3K^2 + 2K(\tau_c - \tau_p)\ \textrm{real symbols}$. The extra $K^2 - K$ fronthaul signaling is due to the way the channel estimation error statistics are shared between the APs as per algorithm described therein. The algorithm presented in Algorithm \ref{Algo3} is more efficient way of implementing the algorithm described in \cite{icc2020}.
\subsection{The RLS Algorithm \cite{rls2}}
One of the competing algorithms is RLS which is a recursive implementation of ZF algorithm. This algorithm is also a special case of \eqref{genSeq}  and hence can be implemented in a sequential cell-free network. We will compare our proposed algorithm with the RLS algorithm in the simulation section. The fronthaul analysis of the RLSE algorithm is as follows:
\begin{enumerate}[(i)]
\item Each AP forwards the computed signal estimate $\widehat{\mathbf{s}}_l$ of the signal $\mathbf{s}$ to the successive AP i.e., from AP $l$ to $l+1,\ l =1,\cdots L-1$. This corresponds to $2K(\tau_c - \tau_p)$ real symbols.
\item A side information hermitian-matrix of size $K\times K$, which corresponds to $K^2$ real symbols.
\end{enumerate}

Table \ref{tab_fronthaul} presents the fronthaul requirements for all the algorithms under two categories, one being the data estimates which are shared for every channel use and the statistical parameters (conditioned on available side information) or channel estimates which are shared once in every coherence block. The data estimates shared is same for all the sequential algorithms described above and they amount to $2K(\tau_c-\tau_p)$ real symbols. It can be observed that the fronthaul requirement increases linearly with increase $L$. From this, we conclude that the proposed OSLP algorithm has lower fronthaul requirement than the centralized processing implementation for large $L$ and also over few competing sequential algorithms i.e., the RLS and the Algorithm \ref{Algo3}.

\begin{table}
	\normalsize
	\centering
	\begin{tabular}{ |p{4cm}||p{4cm}|p{8cm}| }
		\hline
		\multicolumn{3}{|c|}{Fronthaul requirement in each coherence block} \\
		\hline
		Methods/Algorithms&Data Estimates &Statistical Parameters/Channel Estimates\\
		\hline
		Centralized Processing&$2\tau_cNL$    &$0$\\
		\hline
		OSLP&   $2K(\tau_c-\tau_p)$  & $K^2$\\
		\hline
		N-LMMSE &$2K(\tau_c - \tau_p)$ & $2K^2 + K$\\
		\hline
		S-MR    &$2K(\tau_c-\tau_p)$ & $K$\\
		\hline
		RLS \cite{rls2} &   $2K(\tau_c-\tau_p)$  & $K^2$\\
		\hline
	\end{tabular}
	\caption{Summary of fronthaul signaling for various algorithms} \label{tab_fronthaul}
	\vspace{-6mm}
\end{table}
\subsection{{Latency}}\label{SecLatency}
The latency is defined as the amount of time required for data signals received at AP $1$ to reach the CPU. The data signals transmitted per channel use would take $L$ time blocks (each time block corresponds to processing time of each AP) to reach the CPU. However, all APs need not wait to process the next received signal i.e., in sequential processing, when AP $l\in\{2,\cdots,L\}$ is computing the estimate of the received signal at any arbitary time block $t_n$, then AP $1$ to AP $(l-1)$ can process the next payload data received in time blocks $t_{n+1}$ to $t_{n+l-1}$ as depicted in Fig. \ref{fig:Latency}. Thus, if $t_u$ channel uses are allocated for transmission of data, there will be $\tau_u+L-1$ rows in the Fig. \ref{fig:Latency} which amounts to latency of $t_u+L$ time blocks as opposed to $t_u L$ time blocks. The idea here is that the APs can process on next consecutive data as soon as they are available. When $t_u \gg L$, which is practically the case, the latency is approximately $t_u$ and the extra delay of $L$ time blocks is small. Thus, for $t_u$ channel uses the latency is approximately $t_c$ time blocks and hence sequential processing practically have less effect on overall latency. The latency can be further reduced by an alternative OSLP algorithm as described below, which is a semi distributed implementation.
\begin{figure}[!htbp]
\centering
\includegraphics[width=0.6\textwidth]{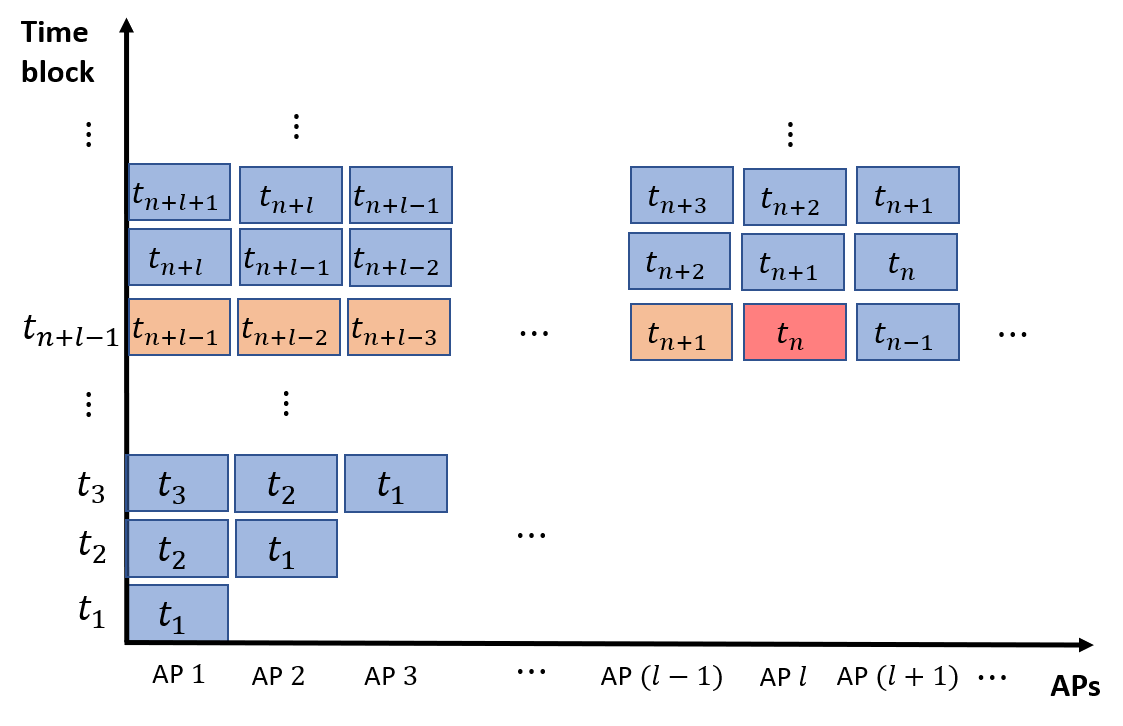}
\caption{{Depiction of processing of data by each AP in each time block.}}
\label{fig:Latency}
\vspace{-12mm}
\end{figure}

To describe the alternative OSLP algorithm (same as the OSLP but different implementation), recall that the estimate of the signal at the CPU with the centralized LMMSE receiver from \eqref{centEst} as
\begin{equation}\label{altOSLP}
\begin{aligned}
\widehat{\mathbf{s}}_L^\textrm{c}= \mathbf{Q}\widehat{\mathbf{G}}_L^H \left(\mathbf{K}_L + \widehat{\mathbf{G}}_L\mathbf{Q}\widehat{\mathbf{G}}_L^H\right)^{-1}\mathbf{z}_L.
\end{aligned}
\end{equation} 
This estimate can be written in an equivalent form as
\begin{equation}
\begin{aligned} \widehat{\mathbf{s}}_L^\textrm{c}&\overset{(a)}{=}\left(\mathbf{Q}^{-1}+\widehat{\mathbf{G}}_L^H\mathbf{K}_L^{-1}\widehat{\mathbf{G}}_L\right)^{-1}\widehat{\mathbf{G}}_L^H\mathbf{K}_L^{-1}\mathbf{z}_L\\
&=\left(\mathbf{Q}^{-1}+\sum_{l=1}^{L}\widehat{\mathbf{H}}_l^H\mathbf{\Sigma}_l^{-1}\widehat{\mathbf{H}}_l\right)^{-1}\left(\sum_{l=1}^{L}\widehat{\mathbf{H}}_l^H\mathbf{\Sigma}_l^{-1}\mathbf{y}_l\right),
\end{aligned}
\end{equation} 
where $(a)$ is an alternative form of the LMMSE receiver matrix \cite{stevenkay}. This approach is similar to the partial decentralization method presented in \cite{JeonPartialDecent} but the authors therein analyzed it for the case of perfect CSI. It is worth noting that the estimate in \eqref{altOSLP} (or equivalently the OSLP estimate) is also the maximum a posteriori (MAP) estimate if the information signal vector is Gaussian distributed. It is interesting to note that this algorithm can be extended to a tree network \cite{Bertilsson}. 

In this alternative form of the OSLP algorithm, after the APs have estimated the channels, they simultaneously compute their quadratic forms $\widehat{\mathbf{H}}_l^H\mathbf{\Sigma}_l^{-1}\widehat{\mathbf{H}}_l$ and then AP $l\in\{1,\cdots,L\}$ forwards the following cumulative sum of quadratic form to AP $(l+1)$ 
\begin{equation}\label{quadForm}
\mathbf{M}_l = \mathbf{M}_{(l-1)}+\widehat{\mathbf{H}}_l^H\mathbf{\Sigma}_l^{-1}\widehat{\mathbf{H}}_l,
\end{equation}
where $\mathbf{M}_0$ is a $K\times K$ matrix with only zeros. When the CPU receives $\mathbf{M}_L$, it computes the inverse matrix in \eqref{altOSLP} (computed once in every coherence block).
Next, in each $\tau_c-\tau_p$ channel uses, all APs simultaneously compute the local weighted MR estimate $\widehat{\mathbf{H}}_l^H\mathbf{\Sigma}_l^{-1}\mathbf{y}_l$ using their corresponding received signal. Then, AP $l\in\{1,\cdots,L\}$ forwards the following cumulative sum of weighted MR estimated signals to the CPU:
\begin{equation}\label{altSMR}
\widetilde{\mathbf{s}}_{l} =\widetilde{\mathbf{s}}_{(l-1)} +  \widehat{\mathbf{H}}_l^H\mathbf{\Sigma}_l^{-1}\mathbf{y}_l,
\end{equation}
where $\widetilde{\mathbf{s}}_{l} = [\widetilde{s}_{1l},\cdots,\widetilde{s}_{Kl}]^T$ is the local weighted MR estimate of the UEs payload with $\widetilde{s}_{kl}$ being the $k$th UE local MR estimate and $\widetilde{\mathbf{s}}_{0} = \mathbf{0}$. Upon receiving $\widetilde{\mathbf{s}}_{L}$, the CPU computes the estimate $\widehat{\mathbf{s}}_{L}^{c}$ as
\begin{equation}\label{altOSLPEst}
\widehat{\mathbf{s}}_{L}^{c} =\left(\mathbf{Q}^{-1}+\mathbf{M}_L\right)^{-1}\widetilde{\mathbf{s}}_{L}.
\end{equation}
With the alternative OSLP algorithm, only \eqref{altSMR} needs to be computed for every channel use and \eqref{quadForm} is computed once in every coherence block, thus helping in reducing the overall latency as it involves only add and forwards mechanism for every channel use. All the results established earlier for the OSLP algorithm holds equally true for the alternative OSLP algorithm as well. The fronthaul requirement for the alternative OSLP is the same as that of original OSLP algorithm, they differ only in the way processing is done. In the former case, final fusion of data is done at the CPU where $K\times K$ matrix inversion is required to obtain the final estimate, while in the later case processing is done at APs where $N\times N$ matrix inversion is computed. But it should be noted that, when the APs are capable of implementing the original OSLP algorithm then it is advantageous over alternative OSLP algorithm mainly because the alternative OSLP algorithm depends on the number of UEs and with increase in the number UEs the computation of quadratic form \eqref{quadForm} at the APs and especially inverting the $K\times K$ matrix in \eqref{altOSLPEst} at the CPU would be computationally expensive. Thus, the original OSLP and the alternative OSLP algorithm trade-off in terms of computational complexity and latency.
\section{Numerical Results and Discussions}\label{numericalResults}
In this section, we evaluate the performance of the proposed OSLP algorithm through numerical results by considering a simulation setup in an area of 125 m $\times$ 125 m. We consider the achievable SE as the performance metrics. We assume that the APs are placed equidistant on a radio stripe of length 500 m which is wrapped around the square perimeter of the simulation area. The propagation model considered for analysis is 3GPP Urban Microcell model \cite[Table~B.1.2.1-1]{LTE2010b} with 2 GHz carrier frequency and the large-scale fading coefficient as
\begin{equation}
\beta_{kl} = -30.5 - 36.7\textrm{log}_{10}\left(\frac{{d}_{kl}}{1\textrm{m}}\right),
\end{equation} 
where $d_{kl}$ is the distance between AP $l$ and UE $k$ (this includes a vertical height difference of 5 m between the APs and the UEs). We assume further that the UEs are uniformly distributed within the concentric square of $100$ m $\times$ $100$ m. Each UE transmits with 50 mW power unless otherwise mentioned. The noise power $\sigma^2$ is $-92$ dBm, the bandwidth is 100 MHz, the coherence block length $\tau_c = 2000$ channel uses, and the number of orthogonal pilot sequences $\tau_p = \textrm{min}(K,20)$. The pilot assignment is done as per the algorithm in \cite{BjorScalableCellfree}. The total number of APs is $L = 24$ and each has $N=4$ antennas unless otherwise stated. The spatial correlation is modeled using the local scattering model \cite[Sec~2.6]{massivemimobook}. We consider a uniform linear array for each AP with half wavelength antennas spacing and the multipath components are Gaussian distributed in the angular domain with a 15 degree standard deviation around the nominal angle to the user. Note that the results presented do not include the alternative OSLP algorithm because it has the same performance as the original OSLP algorithm in Algorithm \ref{Algo1} and all the results presented equivalently hold for both implementations.

In Fig. \ref{figSE}, we plot the CDFs of the SE for users at random locations for the proposed OSLP algorithm and compare with other competing algorithm including centralized scheme. The results demonstrate and validate the claim for the equivalence of the proposed OSLP algorithm with the centralized LMMSE implementation (labelled as "Cent LMMSE" in plots). Fig. \ref{figSE} also demonstrates that the S-MR (sequential MR given in \eqref{SMR_1}) and MR in a centralized setup (labelled as "Cent MR" in plots) are equal. The results plotted are for the case when the number of UEs considered is $K=10$. We also compare these results with Algorithm \ref{Algo3} (labelled "Algo. 2") and we observe that it has superior performance over MR. This is because Algorithm \ref{Algo3} not only makes use of prior knowledge of the channel and the noise statistics but also takes the advantage of APs cooperation in a sequential setup with side information from other APs which help in suppressing the interference more efficiently than MR. 
\begin{figure}[t!]
	\centering
	\includegraphics[width=0.6\textwidth]{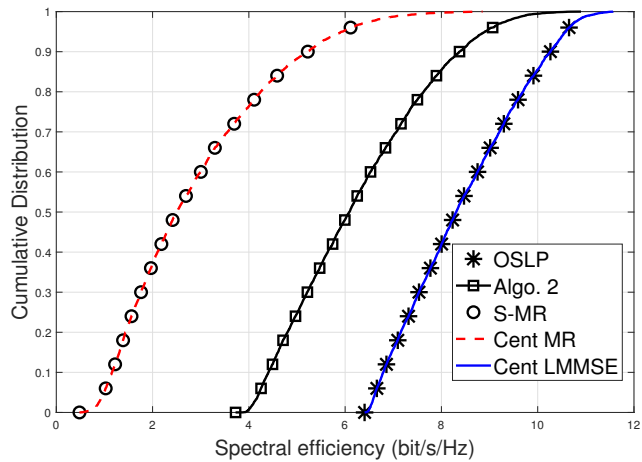}
	\caption{Comparison of the proposed OSLP algorithm with centralized LMMSE scheme $L=24$ and $K=10$ with the CDF of the SE.}
	\label{figSE} \vspace{-8mm}
\end{figure}

In Fig. \ref{figSE_N1}, the proposed OSLP algorithm is compared with the RLS \cite{rls2} algorithm (discussed in the subsection of fronthaul analysis) with the CDF of the SE of the users. Since, \cite{rls2} considered $N=1$, we make the same assumption in Fig. 4 to achieve a fair comparison. The results have been plotted for the system model of this paper with imperfect CSI for $K=20$ and $K=24$. It is observed that as the number of UEs increases the performance gain of the OSLP algorithm over the RLS algorithm increases. For instance, with $K=24$, the proposed OSLP algorithm gains 1.24 bit/sec/Hz as compared to the RLS algorithm with 50\% probability. Interestingly, for $K=24$, with Algorithm \ref{Algo3}, the UEs gain 0.3 bit/s/Hz over the RLS algorithm with 50\% probability. It has to be noted that for the case of $K = 24$, pilot contamination effect is taken into consideration. This shows that the RLS performance is poor when the system is heavily loaded and pilot contamination limited. 
\begin{figure}[t!]
\centering
\includegraphics[width=0.6\textwidth]{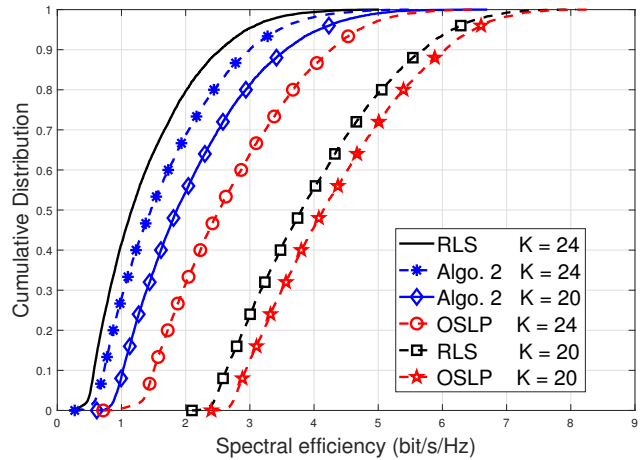}
\caption{Comparison of the proposed OSLP algorithm with other competing algorithms with varying $K$ and fixed $L=24,\ N=1$ using the CDF of the SE.}
\label{figSE_N1} \vspace{-8mm}
\end{figure}

In Fig. \ref{figSEpow1mW}, the CDF of the SE of the UEs is analyzed at low transmit power i.e., 1 mW for each UE. The number of UEs considered for the results in Fig. \ref{figSEpow1mW} is $K=10$. Besides RLS, another recursive algorithm called stochastic gradient descent (SGD) \cite{rls2} with step size 0.02 is considered. The SGD algorithm does not seem to adopt with imperfections in CSI and we observe it has inferior performance. With the proposed OSLP algorithm, the SE of the UEs with 50\% probablity have 0.24 bits/sec/Hz gain over the RLS algorithm. While at high SNR, the RLS method has comparable performance with OSLP (not shown in Fig. \ref{figSEpow1mW}). Since, practical systems operate mostly at low SNR regimes these results illustrates that the RLS and the SGD algorithms have poor performance as compared to the proposed the OSLP algorithm. The RLS algorithm \cite{rls2} is a recursive implementation of centralized ZF algorithm and does not take the side information of the channel and noise statistics into consideration. Hence, it will have inferior performance when compared to the proposed OSLP algorithm in general and especially at low SNR. 
\begin{figure}[t!]
\centering
\includegraphics[width=0.6\textwidth]{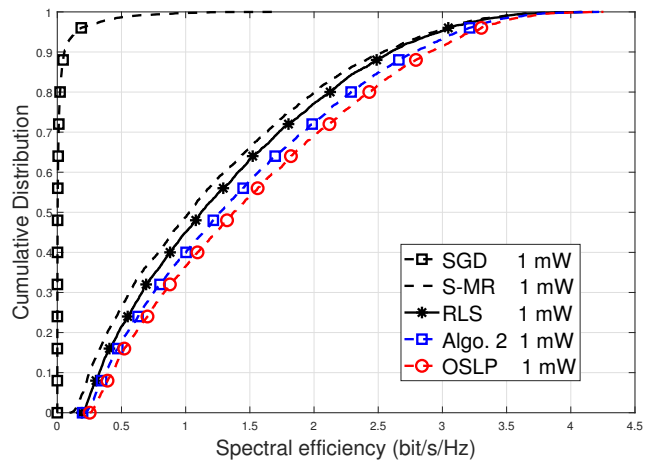}
\caption{The CDFs of the SE for the proposed OSLP algorithm and other competing algorithms at low SNR (with $p_k = 1$mW $\forall k$) for $L=24$ and $N=1$.}
\label{figSEpow1mW} \vspace{-8mm}
\end{figure} 

Finally, in Fig. \ref{figFronthaul}, we illustrate the percentage of fronthaul signaling saved i.e., the number of real symbols saved compared with the centralized processing with a fixed number of UEs and increasing number of APs (since in cell-free mMIMO $L\gg K$). The results in Fig. \ref{figFronthaul} are generated by using  Table \ref{tab_fronthaul} for $K = 20$. As an example with $L = 60$, for the proposed OSLP algorithm, the fronthaul signaling reduces by 90\% as compared to that of the centralized LMMSE algorithm. The RLS algorithm has the same fronthaul signaling requirment as for OSLP.
This analysis concludes that the proposed OSLP algorithm besides being optimal has lower fronthaul requirement as compared to centralized implementation. Since, we have proved analytically that the proposed OSLP algorithm is optimal in the sense of the minimum MSE, we did not include the simulation results due to space constraint and less added insights to the results already presented.
\begin{figure}[!htb]
\centering
\includegraphics[width=0.6\textwidth]{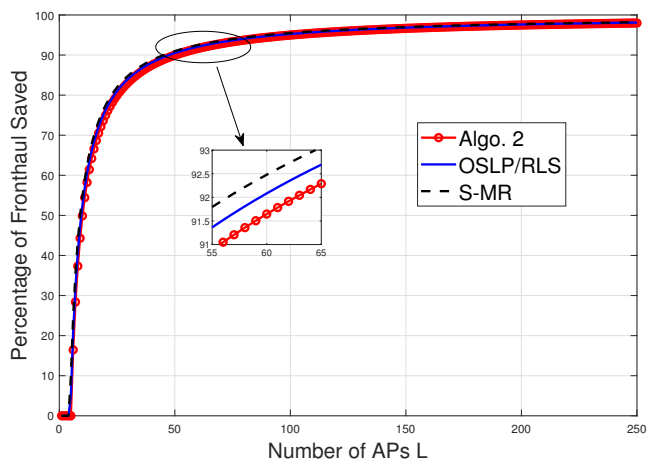}
\caption{Illustration of the percentage of {fronthaul saved} for the proposed OSLP algorithm and centralized implementation.}
\label{figFronthaul} \vspace{-10mm}
\end{figure} 
\section{Conclusion}\label{conclu}
This paper proposes a sequential uplink processing framework that is an optimal choice for any sequential implementation of a cell-free mMIMO networks, for instance radio stripes. We have shown analytically that the proposed OSLP algorithm is optimal in both the maximum SE and the minimum MSE sense. The proposed OSLP algorithm forms the benchmark to analyze the loss of performance of other competing sequential linear algorithms in the sense of SE. We have provided closed-form expressions for the achievable maximum SE and the minimum MSE for the proposed OSLP algorithm, and elaborated on the implications. We also briefly presented an alternative implementation of the same algorithm that is semi-distributed. The main benefit of the OSLP algorithm is that is achieves the same performance as the optimal centralized scheme, but requires much lower fronthaul signaling and makes use of the distributed processors located at the APs.

\section*{Appendix A: Proof of Theorem 1}\label{appendx}
We will prove Theorem \ref{TheoremOSLP} using mathematical induction. To establish that the estimate at AP $L$ using the OSLP algorithm and the estimate obtained using centralized LMMSE receiver are the same, we make use of the alternative expression for the signal estimate obtained by the OSLP algorithm given in \eqref{kalman2cent}. From \eqref{kalman2cent} and \eqref{shat_oslp_cent}, it is sufficient to show that the receive filters are equal: $\overline{\mathbf{V}}_L =$  $\mathbf{V}_L^c$. Besides showing that both receivers are same, we also show simultaneously that the MSE matrix $\mathbf{P}_L$ in the OSLP algorithm is equal to the MSE of centralized scheme i.e., $\mathbf{P}_L = \mathbf{P}_L^{c}$. It then follows that the same SE and MSE are achieved.

Recall that the LMMSE receiver for centralized scheme with $L$ APs, $\mathbf{V}_L^c$, and its corresponding error covariance matrix, $\mathbf{P}_L^c$  for \eqref{centReceiveSig} are given \cite{stevenkay}, respectively, as
\begin{align}
\mathbf{V}_L^\textrm{c} &= \mathbf{Q}\widehat{\mathbf{G}}_L^H \left(\mathbf{K}_L + \widehat{\mathbf{G}}_L\mathbf{Q}\widehat{\mathbf{G}}_L^H\right)^{-1}\label{lmmseReceiver}\\
\mathbf{P}_L^c &= \mathbf{Q} - \mathbf{V}_L^\textrm{c} 
\widehat{\mathbf{G}}_{l}\mathbf{Q}\label{mseMatrix}.
\end{align}
Now we prove the claim with mathematical induction using two cases:
Case ($i$): $L=1$. In the case where there is single AP i.e., only AP 1, then the LMMSE receive matrix at the CPU for centralized scheme is
\begin{equation}\label{Eq:Vc1}
\begin{aligned}
\mathbf{V}_1^\textrm{c} &= \mathbf{Q}\widehat{\mathbf{G}}_1^H \left(\mathbf{K}_1 + \widehat{\mathbf{G}}_1\mathbf{Q}\widehat{\mathbf{G}}_1^H\right)^{-1}\\
&=\mathbf{P}_{0}\widehat{\mathbf{H}}_1^H\left(\mathbf{\Sigma}_1 + \widehat{\mathbf{H}}_1\mathbf{P}_{0}\widehat{\mathbf{H}}_1^H\right)^{-1}\\
& = \overline{\mathbf{V}}_{1},
\end{aligned}
\end{equation}
where the last equality follows from \eqref{vbarL}. Hence, for $L=1$, we have proved that receiver matrices of the two algorithms are same. Next, using \eqref{Eq:Vc1} combined with the MSE expression in \eqref{mseMatrix}, the equivalence of the MSE matrices for both algorithms can be established as follows:
\begin{equation}\label{mseL1}
\begin{aligned}
\mathbf{P}_1^c &= \mathbf{Q} - \overline{\mathbf{V}}_1 
\widehat{\mathbf{G}}_{1}\mathbf{Q}\\
&\overset{(a)}{=}\left(\mathbf{I}-\mathbf{T}_1\widehat{\mathbf{H}}_{1}\right)\mathbf{P}_0\\
&=\mathbf{P}_1.
\end{aligned}
\end{equation}
In \eqref{mseL1}, $(a)$ follows from \eqref{Eq:Vc1} and \eqref{vbarL}. Hence, the estimate $\widehat{\mathbf{s}}_1$ obtained using OSLP algorithm and centralized scheme are equivalent for the case of $L=1$.

Case ($ii$): Assume, that this equivalence of receiver matrices and it's MSE matrices holds for the cases $L \in \{1,\cdots, l\}$ i.e., 
\begin{equation}\label{generL_V}
\begin{aligned}
\mathbf{V}_l^\textrm{c}&=\overline{\mathbf{V}}_l\\
&\overset{(a)}{=}\mathbf{Q}\widehat{\mathbf{G}}_l^H \left(\mathbf{K}_l + \widehat{\mathbf{G}}_l\mathbf{Q}\widehat{\mathbf{G}}_l^H\right)^{-1}\\
&\overset{(b)}{=}\begin{bmatrix}\overline{\mathbf{V}}_{(l-1)} - \mathbf{T}_l\widehat{\mathbf{H}}_{l}\overline{\mathbf{V}}_{(l-1)} & \mathbf{T}_l\end{bmatrix}
\end{aligned}
\end{equation}
and also the MSE matrices
\begin{equation}\label{generL_P}
\begin{aligned}
\mathbf{P}_l^c &=\mathbf{P}_l\\
&\overset{(c)}{=}\mathbf{Q}-\mathbf{V}_l^c \widehat{\mathbf{G}}_{l}\mathbf{Q}\\
&\overset{(d)}{=} \mathbf{Q}-\overline{\mathbf{V}}_l \widehat{\mathbf{G}}_{l}\mathbf{Q}.
\end{aligned}
\end{equation}
In \eqref{generL_V}, $(a)$ follows from \eqref{lmmseReceiver} and $(b)$ is from \eqref{vbarL}. Similarly, $(c)$ in \eqref{generL_P}  is due to \eqref{mseMatrix} and $(d)$ follows from the assumption in \eqref{generL_V}.

We will use these expressions to show that the equivalence holds for the case of $L=(l+1)$ which would complete the proof. For simplicity we let the inverse term in the LMMSE receiver matrix in \eqref{lmmseReceiver} with $L=(l+1)$ APs
{\small
\begin{equation}\label{matInv}
\begin{aligned}
\left(\begin{bmatrix}
\mathbf{K}_l & \mathbf{0}\\
\mathbf{0} & \mathbf{\Sigma}_{(l+1)}
\end{bmatrix} + \begin{bmatrix}
\widehat{\mathbf{G}}_{l}\\
\widehat{\mathbf{H}}_{(l+1)}
\end{bmatrix}\mathbf{Q}\begin{bmatrix}
\widehat{\mathbf{G}}_{l}\\
\widehat{\mathbf{H}}_{(l+1)}
\end{bmatrix}^H\right)^{-1}&=\begin{bmatrix}
\mathbf{A} & \mathbf{B} \\
\mathbf{C} & \mathbf{D}
\end{bmatrix}^{-1}=\begin{bmatrix}
\bar{\mathbf{A}} & \bar{\mathbf{B}} \\
\bar{\mathbf{C}} & \bar{\mathbf{D}}
\end{bmatrix}.
\end{aligned}
\end{equation}
} 
The LMMSE receiver matrix is given by the
\begin{equation}\label{Eq:VcL}
\begin{aligned}
\mathbf{V}_{(l+1)}^\textrm{c}
&=\mathbf{Q}\begin{bmatrix}
\widehat{\mathbf{G}}_{l}\\
\widehat{\mathbf{H}}_{(l+1)}
\end{bmatrix}^H
\begin{bmatrix}
\mathbf{A} & \mathbf{B} \\
\mathbf{C} & \mathbf{D}
\end{bmatrix}^{-1}\\
&=\begin{bmatrix}
\mathbf{Q}\widehat{\mathbf{G}}_{l}^H\bar{\mathbf{A}} + \mathbf{Q}\widehat{\mathbf{H}}_{(l+1)}^H\bar{\mathbf{C}}
& \mathbf{Q}\widehat{\mathbf{G}}_{l}^H\bar{\mathbf{B}} + \mathbf{Q}\widehat{\mathbf{H}}_{(l+1)}^H\bar{\mathbf{D}}\end{bmatrix}\\
&=\begin{bmatrix}\mathbf{F}_1 & \mathbf{F}_2 \end{bmatrix},
\end{aligned}
\end{equation}
where the new variables are defined as
\begin{equation}
\begin{aligned}
\mathbf{A}&= \mathbf{K}_{l} + \widehat{\mathbf{G}}_{l}\mathbf{Q}\widehat{\mathbf{G}}_{l}^H,\\
\mathbf{B}&=\widehat{\mathbf{G}}_{l}\mathbf{Q}\widehat{\mathbf{H}}_{(l+1)}^H,\\
\mathbf{C}&= \widehat{\mathbf{H}}_{(l+1)}\mathbf{Q}\widehat{\mathbf{G}}_{l}^H,\\
\mathbf{D}&= 
\mathbf{\Sigma}_{(l+1)} + \widehat{\mathbf{H}}_{(l+1)}\mathbf{Q}\widehat{\mathbf{H}}_{(l+1)}^H
\end{aligned}
\end{equation}
and
\begin{equation}
\begin{aligned}
\ \mathbf{F}_1 &= \mathbf{Q}\widehat{\mathbf{G}}_{l}^H\bar{\mathbf{A}} + \mathbf{Q}\widehat{\mathbf{H}}_{(l+1)}^H\bar{\mathbf{C}},\\
\mathbf{F}_2 &= \mathbf{Q}\widehat{\mathbf{G}}_{l}^H\bar{\mathbf{B}} + \mathbf{Q}\widehat{\mathbf{H}}_{(l+1)}^H\bar{\mathbf{D}}.
\end{aligned}
\end{equation}
To find the matrix inverse in \eqref{matInv}, we utilize the following block matrix inversion identity \cite{horn2012matrix} 
\begin{equation}\label{matrixInvLemma}
\begin{bmatrix}
\mathbf{A} & \mathbf{B} \\
\mathbf{C} & \mathbf{D}
\end{bmatrix}^{-1} = \begin{bmatrix}
\mathbf{A}^{-1} + \mathbf{A}^{-1}\mathbf{B}\left(\mathbf{D} - \mathbf{CA}^{-1}\mathbf{B}\right)^{-1}\mathbf{CA}^{-1} &
-\mathbf{A}^{-1}\mathbf{B}\left(\mathbf{D} - \mathbf{CA}^{-1}\mathbf{B}\right)^{-1} \\
-\left(\mathbf{D}-\mathbf{CA}^{-1}\mathbf{B}\right)^{-1}\mathbf{CA}^{-1} &
\left(\mathbf{D} - \mathbf{CA}^{-1}\mathbf{B}\right)^{-1}
\end{bmatrix}
\end{equation}
First, we compute $\bar{\mathbf{D}}$ because it simplifies other expressions:
\begin{equation}\label{invTermD}
\begin{aligned}
\bar{\mathbf{D}}&^{-1} = \left(\mathbf{D} - \mathbf{CA}^{-1}\mathbf{B}\right)\\&=\mathbf{\Sigma}_{(l+1)} + \widehat{\mathbf{H}}_{(l+1)}\mathbf{Q}\widehat{\mathbf{H}}_{(l+1)}^H - \widehat{\mathbf{H}}_{(l+1)}\mathbf{Q}\widehat{\mathbf{G}}_{l}^H \mathbf{A}^{-1}\widehat{\mathbf{G}}_{l}\mathbf{Q}\widehat{\mathbf{H}}_{(l+1)}^H\\
&=\widehat{\mathbf{H}}_{(l+1)}\left(\mathbf{Q}- \mathbf{Q}\widehat{\mathbf{G}}_{l}^H \mathbf{A}^{-1}\widehat{\mathbf{G}}_{l}\mathbf{Q}\right)\widehat{\mathbf{H}}_{(l+1)}^H + \mathbf{\Sigma}_{(l+1)}\\
&\overset{(a)}{=}\widehat{\mathbf{H}}_{(l+1)}\left(\mathbf{Q}-\overline{\mathbf{V}}_l \widehat{\mathbf{G}}_{l}\mathbf{Q}\right)\widehat{\mathbf{H}}_{(l+1)}^H + \mathbf{\Sigma}_{(l+1)}\\
&\overset{(b)}{=} \widehat{\mathbf{H}}_{(l+1)}\mathbf{P}_l\widehat{\mathbf{H}}_{(l+1)}^H + \mathbf{\Sigma}_{(l+1)}.
\end{aligned}
\end{equation}
In \eqref{invTermD}, $(a)$ and $(b)$ follows from the assumptions in \eqref{generL_V} and \eqref{generL_P}, respectively. Next, \eqref{invTermD} can be utilized to simplify the variables
$\bar{\mathbf{A}}, \bar{\mathbf{B}},\bar{\mathbf{C}},$ and $\bar{\mathbf{D}}$ as follows:
\begin{equation}\label{invTerms}
\begin{aligned}
\bar{\mathbf{D}}&=(\widehat{\mathbf{H}}_{(l+1)}\mathbf{P}_l\widehat{\mathbf{H}}_{(l+1)}^H + \mathbf{\Sigma}_{(l+1)})^{-1},\\
\bar{\mathbf{A}}&=\mathbf{A}^{-1}  + \mathbf{A}^{-1}\widehat{\mathbf{G}}_{l}\mathbf{Q}\widehat{\mathbf{H}}_{(l+1)}^H\bar{\mathbf{D}}^{-1}\widehat{\mathbf{H}}_{(l+1)}\mathbf{Q}\widehat{\mathbf{G}}_{l}^H\mathbf{A}^{-1},\\
\bar{\mathbf{B}}&=-\mathbf{A}^{-1}\widehat{\mathbf{G}}_{l}\mathbf{Q}\widehat{\mathbf{H}}_{(l+1)}^H\bar{\mathbf{D}}^{-1},\\
\bar{\mathbf{C}}&=-\bar{\mathbf{D}}^{-1}\widehat{\mathbf{H}}_{(l+1)}\mathbf{Q}\widehat{\mathbf{G}}_{l}^H\mathbf{A}^{-1}.
\end{aligned}
\end{equation}
In \eqref{invTerms}, the term $\left(\mathbf{D} - \mathbf{CA}^{-1}\mathbf{B}\right)$ in all the variables is substituted by simplified expression from \eqref{invTermD}. Now that we have obtained all the required variables to compute the inverse term in \eqref{Eq:VcL}, we proceed to first simplify $\mathbf{F}_1$ as
\begin{equation}\label{F1}
\begin{aligned}
\mathbf{F}_1 &=\mathbf{Q}\widehat{\mathbf{G}}_{l}^H\bar{\mathbf{A}} + \mathbf{Q}\widehat{\mathbf{H}}_{(l+1)}^H\bar{\mathbf{C}}\\&=\mathbf{Q}\widehat{\mathbf{G}}_{l}^H\mathbf{A}^{-1}  +\mathbf{Q}\widehat{\mathbf{G}}_{l}^H \mathbf{A}^{-1}\widehat{\mathbf{G}}_{l}\mathbf{Q}\widehat{\mathbf{H}}_{(l+1)}^H\bar{\mathbf{D}}^{-1}\widehat{\mathbf{H}}_{(l+1)}\mathbf{Q}\widehat{\mathbf{G}}_{l}^H\mathbf{A}^{-1} + \mathbf{Q}\widehat{\mathbf{H}}_{(l+1)}^H\bar{\mathbf{C}}\\
&=\overline{\mathbf{V}}_l -  \left(\mathbf{Q}-\overline{\mathbf{V}}_l\widehat{\mathbf{G}}_{l}\mathbf{Q}\right)\widehat{\mathbf{H}}_{(l+1)}^H\bar{\mathbf{D}}^{-1}\widehat{\mathbf{H}}_{(l+1)}^H\overline{\mathbf{V}}_l \\
&=\overline{\mathbf{V}}_l - \mathbf{P}_l\widehat{\mathbf{H}}_{(l+1)}^H\bar{\mathbf{D}}^{-1}\widehat{\mathbf{H}}_{(l+1)}^H\overline{\mathbf{V}}_l\\
&\overset{(a)}{=}\overline{\mathbf{V}}_l -  \mathbf{T}_{(l+1)}\widehat{\mathbf{H}}_{(l+1)}\overline{\mathbf{V}}_l,
\end{aligned}
\end{equation}
where $(a)$ follows from \eqref{Tl} and \eqref{invTermD}. Next, we simplify the remaining term $\mathbf{F}_2$ as
\begin{equation}
\begin{aligned}
\mathbf{F}_2&=\mathbf{Q}\widehat{\mathbf{G}}_{l}^H\bar{\mathbf{B}} + \mathbf{Q}\widehat{\mathbf{H}}_{(l+1)}^H\bar{\mathbf{D}}\\
&=\left(\mathbf{Q}-\mathbf{Q}\widehat{\mathbf{G}}_{l}^H\mathbf{A}^{-1}\widehat{\mathbf{G}}_{l}\mathbf{Q}\right)\widehat{\mathbf{H}}_{(l+1)}^H\bar{\mathbf{D}}^{-1}\\ &=\mathbf{P}_l\widehat{\mathbf{H}}_{(l+1)}^H\bar{\mathbf{D}}^{-1}\\
&= \mathbf{T}_{(l+1)}.
\end{aligned}
\end{equation}
Therefore, \eqref{Eq:VcL} using \eqref{vbarL} becomes
\begin{equation}\label{lmmseL}
\begin{aligned}
\mathbf{V}_{(l+1)}^\textrm{c}&=\begin{bmatrix}\overline{\mathbf{V}}_l - \mathbf{T}_{(l+1)}\widehat{\mathbf{H}}_{(l+1)}\overline{\mathbf{V}}_l & \mathbf{T}_{(l+1)}\end{bmatrix}\\
&=\overline{\mathbf{V}}_{(l+1)}
\end{aligned}
\end{equation}
which is the desired result for the equivalence of the receiver matrices. Moreover, the error covariance equivalence is established as
\begin{equation}\label{mseL}
\begin{aligned}
\mathbf{P}&_{(l+1)}^c = \mathbf{Q} - \mathbf{V}_{(l+1)}^{c} \begin{bmatrix}
\widehat{\mathbf{G}}_{l}\\
\widehat{\mathbf{H}}_{(l+1)}
\end{bmatrix}\mathbf{Q}\\
&\overset{(a)}{=}  \mathbf{Q} - \overline{\mathbf{V}}_l\widehat{\mathbf{G}}_{l} \mathbf{Q} + \mathbf{T}_{(l+1)}\widehat{\mathbf{H}}_{(l+1)}\overline{\mathbf{V}}_l\widehat{\mathbf{G}}_{l} \mathbf{Q}-\mathbf{T}_{(l+1)}\widehat{\mathbf{H}}_{(l+1)} \mathbf{Q}\\
&\overset{(b)}{=}\mathbf{P}_l -  \mathbf{T}_{(l+1)}\widehat{\mathbf{H}}_{(l+1)}\mathbf{P}_l\\
&=\left(\mathbf{I}-\mathbf{T}_{(l+1)}\widehat{\mathbf{H}}_{(l+1)}\right)\mathbf{P}_l\\
&=\mathbf{P}_{(l+1)}.
\end{aligned}
\end{equation}
In \eqref{mseL}, $(a)$ follows from \eqref{lmmseL} and $(b)$ is due to \eqref{generL_P}. Thus, the \eqref{generL_V} and \eqref{generL_P} holds for any general $L$. This concludes the proof.
\bibliographystyle{IEEEtran}
\bibliography{IEEEabrv,reff}



\end{document}